\documentclass[sn-mathphys-num]{sn-jnl}

\usepackage{graphicx}%
\usepackage{multirow}%
\usepackage{amsmath,amssymb,amsfonts}%
\usepackage{amsthm}%
\usepackage{mathrsfs}%
\usepackage[title]{appendix}%
\usepackage{xcolor}%
\usepackage{textcomp}%
\usepackage{manyfoot}%
\usepackage{booktabs}%
\usepackage{algorithm}%
\usepackage{algorithmicx}%
\usepackage{algpseudocode}%
\usepackage{listings}%
\usepackage[capitalize,nameinlink,sort]{cleveref}

\usepackage[textsize=tiny]{todonotes}

\theoremstyle{plain}%
\newtheorem{theorem}{Theorem}
\newtheorem{proposition}[theorem]{Proposition}
\newtheorem{lemma}[theorem]{Lemma}%
\newtheorem{corollary}[theorem]{Corollary}%

\theoremstyle{definition}%
\newtheorem{remark}[theorem]{Remark}%

\newtheorem{definition}[theorem]{Definition}%

\DeclareMathOperator{\Fac}{\mathcal{L}}
\DeclareMathOperator{\N}{\mathbb{N}}
\DeclareMathOperator{\Z}{\mathbb{Z}}

\newcommand{\infw}[1]{\mathbf{#1}}
\newcommand{\bc}[2]{\mathsf{b}_{#1}^{(#2)}}
\newcommand{\ac}[1]{\mathsf{a}_{#1}}
\DeclareMathOperator{\pref}{pref}
\DeclareMathOperator{\rep}{rep}
\DeclareMathOperator{\NE}{ne}
\DeclareMathOperator{\PR}{pr}
\DeclareMathOperator{\fact}{Fac}

\definecolor{lime}{HTML}{A6CE39}
\DeclareRobustCommand{\orcidicon}{
	\begin{tikzpicture}
	\draw[lime, fill=lime] (0,0) 
	circle [radius=0.16] 
	node[white] {{\fontfamily{qag}\selectfont \tiny ID}};
	\draw[white, fill=white] (-0.0625,0.095) 
	circle [radius=0.007];
	\end{tikzpicture}
	\hspace{-2mm}
}
\foreach \x in {A, ..., Z}{\expandafter\xdef\csname orcid\x\endcsname{\noexpand\href{https://orcid.org/\csname orcidauthor\x\endcsname}
			{\noexpand\orcidicon}}
}

\begin{document}

\title{Automatic Abelian Complexities of Parikh-Collinear Fixed Points}


\author{\fnm{Michel} \sur{Rigo}\orcidR{}}\email{m.rigo@uliege.be}

\author{\fnm{Manon} \sur{Stipulanti}\orcidS{}}\email{m.stipulanti@uliege.be}

\author{\fnm{Markus A.} \sur{Whiteland}\orcidW{}}\email{mwhiteland@uliege.be}

\affil{\orgdiv{Department of Mathematics}, \orgname{ULi\`ege}, \orgaddress{\street{All\'ee de la D\'ecouverte 12}, \city{Li\`ege}, \postcode{4000},  \country{Belgium}}}


\abstract{
Parikh-collinear morphisms have the property that all the Parikh vectors of the images of letters are collinear, i.e., the associated adjacency matrix has rank~$1$. In the conference DLT-WORDS 2023 we showed that fixed points of Parikh-collinear morphisms are automatic. We also showed that the abelian complexity function of a binary fixed point of such a morphism is automatic under some assumptions.
In this note, we fully generalize the latter result. Namely, we show that the abelian complexity function of a fixed point of an arbitrary, possibly erasing, Parikh-collinear morphism is automatic. Furthermore, a deterministic finite automaton with output generating this abelian complexity function is provided by an effective procedure. To that end, we discuss the constant of recognizability of a morphism and the related cutting set.
}




\maketitle

\section{Introduction}
This paper is an extension of the results in our previous work \cite{RSW-DLT23} that was presented during the joint DLT-WORDS 2023 conference. The main objects of interest are fixed points of Parikh-collinear morphisms which are defined as follows. It is assumed that the alphabet $A=\{a_1<\cdots<a_k\}$ is ordered and $\Psi(w)$ denotes the \emph{abelianization} or \emph{Parikh vector} $(|w|_{a_1},\ldots,|w|_{a_k})$ counting the number of different letters constituting the word $w\in A^*$. A morphism $f \colon A^* \to B^*$ is \emph{Parikh-collinear} if the Parikh vectors $\Psi(f(b))$, $b\in A$, are collinear (or pairwise $\Z$-linearly dependent).

Parikh-collinear morphisms have received some attention in recent years. The authors of
\cite[Sec.~4]{AlloucheDekQue2021} list a dozen of fixed points of Parikh-collinear morphisms appearing in the OEIS~\cite{Sloane}, e.g., \href{https://www.oeis.org/A285249}{A285249}. Cassaigne {\em et al.}~characterized Parikh-collinear morphisms as those morphisms that map all words to words with bounded abelian complexity \cite{CassaigneRichomeSaariZamboni2011}. These morphisms also provide infinite words with interesting properties with respect to the so-called $k$-binomial equivalence~$\sim_k$. Two words $u, v\in A^*$ are \emph{$k$-binomially equivalent} if $\binom{u}{x} = \binom{v}{x}$, for all $x\in A^*$ with $|x|\le k$. Recall that a binomial coefficient $\binom{u}{x}$ counts the number of times $x$ occurs as a subword of $u$. The \emph{$k$-binomial complexity function} of an infinite word~$\infw{x}$ introduced in \cite{RigoSalimov2015} is defined as $\bc{\infw{x}}{k} \colon \N \to \N$, $n\mapsto \#(\Fac_n(\infw{x})/{\sim_k})$, i.e., length-$n$ factors in~$\infw{x}$ are counted up to $k$-binomial equivalence. (Here $\bc{\infw{x}}{1}$ is the usual abelian complexity function \cite{Erdos1957}.) For a survey on abelian properties of words, see \cite{FiciPuzynina}. In a recent work, we showed that a morphism is Parikh-collinear if and only if it maps all words with bounded $k$-binomial complexity to words with bounded $(k+1)$-binomial complexity (for all $k$) \cite{RSW2}. Thus each fixed point of a Parikh-collinear morphism has a bounded $k$-binomial complexity for all $k$ (and in particular a bounded abelian complexity).

Let us summarize the contributions from~\cite{RSW-DLT23} which are twofold. To explain them, we assume the reader is familiar with automatic sequences (see, e.g., \cite{AS,Shallit2022logical}) and we introduce some terminology.
For an arbitrary morphism $\sigma \colon A^* \to A^*$, we let $M_\sigma\in\mathbb{N}^{A\times A}$ denote its {\em adjacency matrix},  where $[M_\sigma]_{b,c}=|\sigma(c)|_b$ for all $b,c\in A$. A letter $a \in A$ is called \emph{mortal} if $\sigma^n(a) = \varepsilon$ for some $n \geq 1$. 
If $a$ is not mortal, we call it \emph{immortal}.

\begin{lemma}\label{lem:eigenvaluevector}
Let $f \colon A^* \to A^*$ be Parikh-collinear and $a \in A$ be immortal.
Then $\Psi(f(a))$ is an eigenvector of $M_f$ associated with the eigenvalue
$\sum_{b \in A}|f(b)|_b$.
\end{lemma}

When speaking of \emph{the eigenvalue} of a Parikh-collinear morphism $f$, we mean the eigenvalue
$\sum_{b \in A}|f(b)|_b$ of $M_f$; this is justified as, the matrix $M_f$ having rank $1$, the only other eigenvalue is $0$ with multiplicity $\#A-1$.

Even though Parikh-collinear morphisms are generally non-uniform, we proved the following result.
\begin{theorem}[{\cite[Thm.~5]{RSW-DLT23}}]\label{thm:kautomatic}
Let $f\colon A^* \to A^*$ be a Parikh-collinear morphism prolongable on a letter $a \in A$.  Then the fixed point $f^{\omega}(a)$
is $k$-automatic for the eigenvalue $k$ of $f$.  Furthermore,  a coding together with a $k$-uniform morphism generating the infinite word can be effectively computed.
\end{theorem}

The latter result can be considered folklore: it can be seen as a consequence of \cite[Thm.~2.2 or~4.2]{AlloucheDekQue2021}, the former of which is itself a reformulation of a result of
Dekking \cite{Dekking1978} (we note however, that the statements speak of non-erasing morphisms). It is well known that there exist infinite
sequences that are the fixed points of non-uniform morphisms, but not $k$-automatic for any $k$, and that every $k$-automatic sequence is the image of a fixed point of a non-uniform morphism \cite{ASNonuniform}. A recent preprint \cite{KrawczykM2023automaticity} completely characterizes those uniformly recurrent (i.e., every factor occurs infinitely often and with bounded gaps) morphic words that are automatic. 

Next we proved in \cite[Thm.~10]{RSW-DLT23} that under some mild assumptions (about the automaticity of the cutting set that we will discuss further) the abelian complexity of a binary fixed point of a Parikh-collinear morphism is automatic. We can therefore use an automatic procedure to test whether or not this function is ultimately periodic, for example. Answering a question raised by Salo and Sportiello, considering the abelian complexity of the fixed point $\infw{w}=0100111001\cdots$ of the morphism $f \colon \{0,1\}^* \to \{0,1\}^*$ given by $0 \mapsto 010011$, $1 \mapsto 1001$, we showed that its abelian complexity is aperiodic. 
We also gave a proof sketch showing that the abelian complexity function of a fixed point of a non-erasing Parikh-collinear morphism is automatic.

\subsection{Our contributions}
For this special edition we did not want to replicate the results of the proceedings \cite{RSW-DLT23}. Therefore our main contribution is to generalize \cite[Thm.~10]{RSW-DLT23} to an arbitrary Parikh-collinear morphism: meaning on an alphabet of arbitrary size and the morphism may be erasing.

\begin{theorem}\label{thm:main}
Let $f \colon A^* \to A^*$ be a Parikh-collinear morphism prolongable on the letter $a$.
The the abelian complexity function $\ac{\infw{x}}$ of $\infw{x}:=f^{\omega}(a)$ is $k$-automatic for the eigenvalue $k$ of $f$.
Moreover, the automaton generating $\ac{\infw{x}}$ can be effectively computed given $f$ and $a$.
\end{theorem}

Before proving this result in~\cref{sec: proof main thm}, we first need a computable bound on the so-called recognizability constant. In \cref{sec:2}, we have extracted from \cite{Durand1998,DL2017,DHS1999} and, in particular \cite{beal_perrin_restivo_2023}, the relevant definitions and important results showing that, for our study, such a constant exists. Expressed roughly, when we look at a sufficiently long factor, there is a unique pre-image by the morphism $f$ and there is only one way to factorize this factor using blocks of the form $f(b)$,  where $b$ is a letter.

On this basis, we define in \cref{sec:3} the notion of a cutting set. Since the infinite word $\infw{x}=x_0x_1\cdots$ can be factorized as $f(x_0)f(x_1)\cdots$, this set consists of the integers $|f(x_0\cdots x_j)|$ for all $j\ge 0$. Our main observation is that, for a Parikh-collinear morphism~$f$, this set is $k$-definable. We insist that this is a major element which then enables us to apply a decision procedure about the abelian complexity of $\infw{x}=f^\omega(a)$. Such a procedure is described in \cref{sec:4}. We consider the Parikh-collinear morphism  $0\mapsto 012$, $1\mapsto 112002$, $2\mapsto \varepsilon$ and prove with the help of {\tt Walnut} that the fixed point starting with $1$ has an ultimately periodic abelian complexity $135(377)^\omega$.

\section{On the recognizability}\label{sec:2}
For an arbitrary morphism $\sigma\colon A^* \to A^*$, we define 
\[ |\sigma|:=\max\{|\sigma(b)|\colon b\in A\}\text{ and } \langle \sigma\rangle:=\min\{|\sigma(b)|\colon b\in A\},\]
where,  for a word  $w\in A^*$,  we let $|w|$ denote its length.

In what follows, $f\colon A^* \to A^*$ is a Parikh-collinear morphism prolongable on the letter $a \in A$. For an arbitrary morphic word $\infw{x}$, thanks to \cite{Durand2013,Mitrofanov2011}, one can decide whether $\infw{x}$ is ultimately periodic. In our case thanks to \cref{thm:kautomatic}, $\infw{x}=f^\omega(a)$ is also $k$-automatic, and we can therefore make use of the logical characterization of $k$-automatic sequences and it can be readily decided with {\tt Walnut}~\cite{Mousavi2016automatic} using a formula such as
\begin{equation}
  \label{eq:periodic}
  \neg(\exists p> 0)(\exists i\ge 0)(\forall n\ge i)(\infw{x}(n)=\infw{x}(n+p)).
\end{equation}

We will therefore assume in what follows that $\infw{x}$ is not ultimately periodic. Also, we restrict the alphabet~$A$ to the letters appearing in $f^n(a)$ for some $n$. As an example, for the Parikh-collinear morphism $f:1\mapsto 12,\ 2\mapsto 21,\ 3\mapsto 12$ prolongable on $1$, we consider the restriction to the alphabet $\{1,2\}$.

In what follows, an arbitrary morphism $\sigma$ is called \emph{primitive} if $M_\sigma^n$ only contains positive entries for some $n \in \N$.
\begin{lemma}\label{lem:primitive}
  A non-erasing Parikh-collinear morphism is primitive.
\end{lemma}

\begin{proof}
  Observe that all entries in the adjacency matrix are positive. 
\end{proof}

\begin{remark}
Note that for a non-erasing Parikh-collinear morphism $g$,  we may apply \cite[Thm.~4]{DL2017} which directly provides a computable upper bound on the recognizability constant for the aperiodic word $g^\omega(a)$.
\end{remark}

Since $f$ is Parikh-collinear and possibly erasing,  there is a strong dichotomy among the letters of the alphabet.  Either they are immortal and their image by $f$ contains all letters or, their image by $f$ is empty.  Formally,  for all $b\in A$, either $\Psi(f(b))=0$ or $\Psi(f(b))$ is a non-zero rational multiple of $\Psi(f(a))$. In the latter case, for all $n\ge 0$, $\Psi(f^n(b))$ is therefore non-zero. So, the alphabet is partitioned as $A=B\cup C$ where
\begin{equation}
  \label{eq:alphabets}
B:=\{b\in A\mid f^n(b)\neq\varepsilon,\ \forall n\ge 0\}
\quad\text{ and }\quad
C:=\{b\in A\mid f(b)=\varepsilon\}.  
\end{equation}

\begin{definition}
We use notation from~\cref{eq:alphabets}. Let $\kappa:A^*\to B^*$ be a morphism such that $\kappa(b)=b$ if $b\in B$ and $\kappa(c)=\varepsilon$ for all $c\in C$. 
Now we define a morphism $g:B^*\to B^*$ such that $g(b)=\kappa(f(b))$ for all $b\in B$.
\end{definition}
 Roughly, the image by $g$ of an immortal letter $b$ of $f$ is obtained by deleting the mortal letters appearing in $f(b)$.

The next statement is obvious.

\begin{lemma}\label{lem:fixedpoint}
  With the above notation, $g=\kappa\circ f$ is a non-erasing Parikh-collinear morphism prolongable on $a$ and satisfying $f(g^\omega(a))=f^\omega(a)$.
\end{lemma}

Consider $f\colon 0\mapsto 012$, $1\mapsto 112002$; $2 \mapsto \varepsilon$, we get $g:\{0,1\}^*\to\{0,1\}^*$ such that $g(0)=01$ and $g(1)=1100$.

An infinite word is called \emph{recurrent} if each of its factors appears infinitely often.

\begin{definition}
  Let $\infw{z}$ be a recurrent infinite word and $u$ be a factor of $\infw{z}$. A {\em return word} to $u$ is a non-empty factor $w$ of $\infw{z}$ such that $wu$ contains exactly two occurrences of $u$ as a prefix and as a suffix of $wu$. The infinite word~$\infw{z}$ is {\em $K$-linearly recurrent} if, for all factors $u$, any return word~$w$ to $u$ is such that $|w|\le K |u|$.
\end{definition}

We recall a result from \cite{Durand1998} and \cite[Prop.~12]{DL2017}. It is important to note that the given upper bound is computable.

\begin{proposition}\label{pro:linearlyrecurrent}
  Let $\sigma:A^*\to A^*$ be a primitive morphism prolongable on $a$. The infinite word $\sigma^\omega(a)$ is $K_\sigma$-linearly recurrent and the constant $K_\sigma$ is bounded by $|\sigma|^{4(\# A)^2}$.
\end{proposition}

By the above result and \cref{lem:primitive,lem:fixedpoint}, there exists a constant $K_g$ such that $\infw{y}=g^\omega(a)$ is $K_g$-linearly recurrent.

\begin{corollary}\label{cor:Kf}
  The infinite word $\infw{x}=f^\omega(a)=f(\infw{y})$ is $K_f$-linearly recurrent and the constant $K_f$ is bounded by $K_g |f|/\langle f|_B\rangle$. 
\end{corollary}

\begin{proof}
 Since $\infw{x}=f(\infw{y})$ by~\cref{lem:fixedpoint} the conclusion follows: Let $u$ be a factor of $\infw{x}$. There exists a factor $v$ of $\infw{y}$ such that $f(v)=pus$ for some words $p,s$ of minimal length and $|v|\le |u|/\langle f|_B\rangle$ (recall that the letters of $B$ do not vanish under $f$). Since $\infw{y}$ is linearly recurrent, any return word $r$ to $v$ has length at most $K_g |v|$. Observe that $f(r)$ contains a return word to $u$ and has length bounded above by $K_g |v|\, |f| \leq K_g\frac{|f|}{\langle f|_B \rangle} |u|$.
\end{proof}

We recall the following result \cite[Thm.~24]{DHS1999}.

\begin{proposition}\label{pro:powerfree}
  A $K$-linearly recurrent aperiodic word is $(K+1)$-power-free.
\end{proposition}

The constant of recognizability is usually presented in the framework of shift spaces whose elements are biinfinite words, i.e., sequences indexed by $\mathbb{Z}$. We recap some of the main definitions and results.

\begin{definition}
The {\em shift operator} $S:A^{\mathbb{Z}}\to A^{\mathbb{Z}}$ is defined by $\infw{z}=(z_n)_{n\in\mathbb{Z}} \mapsto S(\infw{z})=(z_{n+1})_{n\in\mathbb{Z}}$. A {\em shift space} is a subset $X\subset A^\mathbb{Z}$ that is shift-invariant, i.e., $S(X)=X$, and topologically closed. The {\em language} of $X$ is the set denoted by $\mathcal{L}(X)$ of factors of the words in $X$. A shift space is {\em aperiodic} if all its elements are aperiodic. Recall that $\infw{z}\in A^\mathbb{Z}$ is {\em periodic} if $\infw{z}=S^n(\infw{z})$ for some $n\ge 1$. 
\end{definition}

Let $\sigma:A^*\to A^*$ be a morphism. We let
\[\Fac(\sigma)=\bigcup_{n\ge 0} \bigcup_{a\in A} \fact(\sigma^n(a))\]
and the so-called {\em substitution shift} associated with~$\sigma$ is 
\[\mathsf{X}(\sigma) =\{ x\in A^\mathbb{Z} \colon \Fac(x)\subset\Fac(\sigma)\}. \]
From the definition, it is clear that $\Fac(\mathsf{X}(\sigma))\subset \Fac(\sigma)$. A morphism $\sigma$ is {\em aperiodic} if the shift space $\mathsf{X}(\sigma)$ is aperiodic.

\begin{proposition}
  Let $f$ be a Parikh-collinear morphism prolongable on a letter $a \in A$  such that $\infw{x}=f^\omega(a)$ is aperiodic. 
 Then the shift space $\mathsf{X}(f)$ is aperiodic.
\end{proposition}

\begin{proof}
  Proceed by contradiction and assume that $\mathsf{X}(f)$ contains a periodic element $u^\infty=\cdots uu\cdot uu\cdots$. Hence there exists $b\in A$ such that $u^{K_f+1}$ is a factor of $f^m(b)$ for some $m$. For large enough $n$, $b$ occurs in $f^n(a)$. Hence $u^{K_f+1}$ occurs in $f^{n+m}(a)$ and thus in $\infw{x}$, which is $K_f$-linearly recurrent by \cref{cor:Kf} and aperiodic by assumption. By \cref{pro:powerfree}, $\infw{x}$ is thus $(K_f+1)$-power-free, a contradiction.
\end{proof}

The notion of return words and linear recurrence naturally extends to shift spaces.

\begin{definition}
  Let $X$ be a shift space and $u\in\Fac(X)$. A non-empty word $w\in\Fac(X)$ is a {\em return word to $u$ in $X$} if $wu\in\Fac(X)$ contains exactly two occurrences of $u$ as a prefix and as a suffix of $wu$. The shift space $X$ is {\em $K$-linearly recurrent} if it is minimal (for every closed stable subset $Y$ of $X$, i.e., $S(Y)\subset Y$, one has $Y=\emptyset$ or $Y=X$) and for all non-empty words $u\in\Fac(X)$, the length of every return word to $u$ in $X$ is bounded by $K|u|$.
\end{definition}

\begin{proposition}\label{pro:Kflin}
  Let $f$ be a Parikh-collinear morphism prolongable on a letter $a \in A$  such that $\infw{x}=f^\omega(a)$ is aperiodic. 
 Then the shift space $\mathsf{X}(f)$ is $K_f$-linearly recurrent.
\end{proposition}

\begin{proof}
  Let $X=\mathsf{X}(f)$, $u\in\mathcal{L}(X)$ and $w$ be a return word to $u$ in $X$. Observe that $\mathcal{L}(X)\subset \mathcal{L}(f)=\mathcal{L}(f^\omega(a))$.
Hence $u,wu$ are factors of $\infw{x}$ which is $K_f$-linearly recurrent by \cref{cor:Kf}.
\end{proof}

We are now ready to first define the notion of recognizable morphism on $X$,
then to introduce recognizable morphism on $X$ with some constant of recognizability.

\begin{definition}
  Let $X\subset A^\mathbb{Z}$ be a shift space. A morphism $\sigma:A^*\to B^*$ is
{\em recognizable on $X$} if, for all $y\in\sigma(X)$, there exists exactly one pair $(x,\ell)\in X\times\mathbb{N}$ such that $0\le \ell < |\sigma(x_0)|$ and $y = S^\ell (\sigma(x))$.
\end{definition}

B\'eal, Perrin,  and Restivo generalized Moss\'e's theorem \cite[Thm.~5.4]{beal_perrin_restivo_2023}.

\begin{theorem}
  Every morphism $\sigma: A^*\to A^*$ is recognizable on $\mathsf{X}(\sigma)$ for aperiodic elements.  In particular, if $\sigma$ is aperiodic, then it is recognizable on $\mathsf{X}(\sigma)$.
\end{theorem}

Let $X$ be a shift space and $u,v$ be two finite words such that $uv\in\Fac(X)$.  The {\em cylinder} with basis $(u,v)$ is defined as 
\[ [u\cdot v]_X = \{\infw{z} \in X \colon \infw{z}_{[-|u|,|v|-1]} = uv\} .\]
In particular, if $u=\varepsilon$, we simply write $[v]_X= \{\infw{z} \in X \colon \infw{z}_{[0,|v|-1]} = v\}$.

\begin{definition}\label{def:recognizability_constant}
  Let $\sigma:A^*\to B^*$ be a morphism. Let $X$ be a shift space on $A$ and let $Y$ be
the closure of $\sigma(X)$ under the shift. A pair $(u,v)$ of words such that $uv \in \Fac(Y)$
is {\em synchronizing} if there is at most one pair $(b,\ell)$ with $b\in A$ and $0\le \ell < |\sigma(b)|$  such that $[u\cdot v]_Y \cap S^\ell \sigma([b]_X )\neq\emptyset$. The morphism $\sigma$ is recognizable on X {\em with constant} $n$ if and only if every pair $(u, v)\in\Fac_n(Y)\times \Fac_{n+1} (Y)$ such that $uv\in\Fac(Y)$ is synchronizing.
\end{definition}

 Let $X$ be an aperiodic shift space. The {\em repetition index} of $X$ (also called {\em critical exponent} in the case of an infinite word) denoted by $\rep(X)$ is the supremum of the set of rational numbers~$e$ such that $\mathcal{L}(X)$ contains words of exponent $e$. Finally, we invoke the following result from \cite{BDP2024}.

\begin{theorem}\label{thm:constrec}
The constant of recognizability on $\mathsf{X}(\sigma)$ of an aperiodic morphism~$\sigma$ is bounded by \(4 \rep(\mathsf{X}(\sigma))\, \ell^2\, |\sigma|^{(2\ell+1)(2+|\sigma|^{(2^\ell+1)\ell})}\) where $\ell=\# A$.
\end{theorem}

We recall the following from \cite{Durand1998} (see also \cite{DHS1999}). This is the extension to shift spaces of \cref{pro:powerfree}: Let $X$ be an aperiodic shift space and suppose that it is $K$-linearly recurrent, then the repetition index is bounded by $\rep(X)<K+1$. Now an immediate application of \cref{pro:linearlyrecurrent,pro:Kflin} together with \cref{thm:constrec} leads to the following result.

\begin{corollary}\label{cor:cnstnt-rec}
  Let $f$ be a Parikh-collinear morphism prolongable on a letter $a \in A$  such that $f^\omega(a)$ is aperiodic. 
The constant of recognizability on $\mathsf{X}(f)$ of the aperiodic morphism~$f$ is bounded by  $4 (|f|^{4\ell^2}+1)\, \ell^2\, |f|^{(2\ell+1)(2+|f|^{(2^\ell+1)\ell})}$ where $\ell=\# A$.
\end{corollary}

Now consider the right-infinite word $\infw{x}=f^\omega(a)$. It appears as a factor of a element in $\mathsf{X}(f)$. Indeed, since $f$ is Parikh-collinear, there exists some $j\ge 1$ such that $f^j(a)=auav$ (one can take $j=2$). Take the sequence $(f^n(au)\cdot f^n(a) f^n(v))_{n\ge 0}$. By compactness, we can extract a subsequence converging to some biinfinite word $z\cdot f^\omega(a)$ belonging to $\mathsf{X}(f)$.

We have done all this to ensure that there is a computable bound~$C$ guaranteeing that the word~$\infw{x}$ is recognizable: there is a window size bounded by $C$ such that any factor within such a window is uniquely ``desubstituted''. More precisely, this will permit us to uniquely detect elements of the cutting set.

\section{The cutting set}\label{sec:3}

We assume that the reader is familiar with $k$-definable sets and $k$-synchronized sequences. For a reference, see \cite{Shallit2022logical,AS}.

Let $\sigma$ be a morphism prolongable on $a$ and write $\infw{x}=\sigma^\omega(a)$. For all $n\ge 0$, we let $\pref_n(\infw{x})$ be the length-$n$ prefix of $\infw{x}$.
The corresponding {\em cutting set} is defined by
\begin{equation}
\label{eq: cutting set}
\mathsf{CS}_{\sigma,a}:=\left\{ |\sigma(\pref_n(\infw{x}))| \colon \ n\ge 0\right\}.
\end{equation}
This set simply provides the indices where blocks $\sigma(b)$,  with $b\in A$,  start in a factorization of $\infw{x}$ of the form $\sigma(x_0)\sigma(x_1)\sigma(x_2)\cdots$.  For example,  applied to the Parikh-collinear morphism $f\colon 0 \mapsto 012$, $1\mapsto 112002$, $2 \mapsto \varepsilon$ considered before,  we get
\[\infw{x}=|012|112002|112002|112002|012|012|\cdots \text{ and }\mathsf{CS}_{f,0}=\{0,3,9,15,21,24,27,\ldots\}.\]
The unary predicate $\mathsf{CS}_{\sigma,a}(n)$ holds true whenever $n\in \mathsf{CS}_{\sigma,a}$.

\begin{proposition}\label{pro:cskdef}
   Let $f$ be a Parikh-collinear morphism prolongable on a letter $a \in A$  such that $\infw{x}=f^\omega(a)$ is aperiodic. Let $k$ be the eigenvalue of $f$. The cutting set $\mathsf{CS}_{f,a}$ is a $k$-definable unary predicate.
\end{proposition}

\begin{proof}
  By \cref{cor:cnstnt-rec}, there exists a constant of recognizability $C$ on $\mathsf{X}(f)$ with the following property. 
  By \cref{def:recognizability_constant},  each factor $w=ucv$ of $\infw{x}=f^\omega(a)$ of length $2C+1$ (here $|u|=|v|=C$ and $c\in A$) gives rise to a synchronizing pair,  i.e.,  there exists a unique pair $(b,\ell)$ where $b\in A$, $0\le \ell<|f(b)|$ such that  $[u\cdot cv]_{\mathsf{X}(f)} \cap S^\ell f([b]_{\mathsf{X}(f)} )\neq\emptyset$. If $\ell=0$, we have detected an element of the cutting set starting at the ``center''~$c$ of the factor~$w$. So with each factor $w$ of length $2C+1$, we associate a Boolean $T(w)$ stating whether or not the center of $w$ belongs to the cutting set. 

  Since $\infw{x}$ is $k$-automatic (see \cref{thm:kautomatic}), for every factor $w$ of length $2C+1$ and all $n\ge C$, the unary formula~$\varphi_w(n)\equiv \infw{x}[n-C,n+C]=w$ tells whether or not $w$ occurs in $\infw{x}$ as a factor centered at position~$n$ (in other words,  whether the position $n$ is the center of $w$ in $\infw{x}$). If $n\ge C$, the formula
  \[\bigvee_{w\in\Fac_{2C+1}(\infw{x})} \varphi_w(n)\wedge T(w)\]
  holds true whenever $n$ belongs to the cutting set. 
  For $n<C$, this can be defined by direct inspection: there is a finite number of elements in $\mathsf{CS}_{f,a}\cap \{0,\ldots,C-1\}$ to encode manually into the final formula.

  Now we have to effectively list all factors of length $2C+1$ occurring in $\infw{x}$.  Since $\infw{x}$ is $k$-automatic, we can effectively get a $k$-uniform morphism $g$ and a coding $\tau$ such that $\infw{x}$ is of the form $\tau(g^\omega(e))$ for some letter $e$. We can first list all length-$2$ factors occurring in $g^\omega(e)=\infw{y}$. For instance, we can use a formula such as $(\exists n)(\infw{y}(n)=b\wedge \infw{y}(n+1)=c)$ to test whether or not $bc$ occurs in $\infw{y}$.  Second,  every factor of length~$2C+1$ of $\infw{y}$ appears in $g^j(bc)$ for some letters $b,c$ and $j=\lceil \log_k (2C+1)\rceil$.  So scanning these words $g^j(bc)$ with a window of size $2C+1$, we get all desired factors and we apply $\tau$ to them to get all factors of length $2C+1$ occurring in~$\infw{x}$.
\end{proof}

Let us observe that the above proposition can be given in a different framework where we focus on the recognizability of a single infinite word with respect to the considered morphism. We take the following definition from \cite{DL2017} adapted to (right) infinite words.

\begin{definition}
  Let $\infw{x}=x_0x_1 \cdots$ be a fixed point of a prolongable morphism $\sigma$. We say that $\sigma$ is \emph{recognizable} on~$\infw{x}$ if there exists a constant $C>0$ such that for all $n\ge 0$ and all $i$ such that $|\sigma(x_0\cdots x_{i-1})|
  \ge C$,  we have
  \begin{eqnarray*}
    \infw{x}[n-C,n+C]=\infw{x}[|\sigma(x_0\cdots x_{i-1})|-C,|\sigma(x_0\cdots x_{i-1})|+C]\\
    \Rightarrow
    (\exists j)(n=|\sigma(x_0\cdots x_{j-1})|\wedge x_i=x_j).
  \end{eqnarray*}
  The least $C$ with the above property is then called the \emph{constant of recognizability} of $\sigma$ on~$\infw{x}$. 
\end{definition}

The reader may readily adapt the proof of \cref{pro:cskdef} to the following situation.

\begin{theorem}\label{thm:more_gen}
  Let $\infw{x}=\sigma^\omega(a)$ be a fixed point of a prolongable morphism~$\sigma$. If $\sigma$ is recognizable on $\infw{x}$ with computable recognizability constant $C$ and if $\infw{x}$ is $k$-automatic, then  the cutting set $\mathsf{CS}_{\sigma,a}$ is a $k$-definable unary predicate.
\end{theorem}

Given an integer $i$, we look for two consecutive integers around $i$,  the next and previous elements found in $C$. The next lemma is obvious (a proof can be found in \cite{RSW-DLT23}).

\begin{lemma}
\label{lem: ne and pr}
  Let $C=\{0=c_0<c_1<c_2<\cdots\}$ be an infinite $k$-definable subset of $\N$. The functions $\NE:\N\to\N$ mapping $i$ to the least element in $C$ greater than or equal to~$i$ and $\PR:\N\to\N$ mapping $i$ to the greatest element in $C$ less than $i$, are $k$-definable. (We set $\PR(0)=0$.)
\end{lemma}

\section{Proof of \texorpdfstring{\cref{thm:main}}{Theorem 3}}
\label{sec: proof main thm}
We are now in the position to prove the main theorem.
The following theorem of Shallit is crucial in the final step.

\begin{theorem}[\cite{Shallit2021abelian}]
\label{thm:shallit}
Let $\infw{x}$ be an automatic sequence, and assume that
\begin{enumerate}
\item \label{it:1stcond} the sequence $(\Psi(\pref_n(\infw{x}))_{n\geq 0}$ is synchronized; and
\item \label{it:2ndcond} the abelian complexity function $\ac{\infw{x}}\colon \N \to \N$ is bounded above by a constant.
\end{enumerate}
Then $(\ac{\infw{x}}(n))_{n\geq 0}$ is an automatic sequence and the DFAO computing it
is effectively computable.

Furthermore, if Condition \ref{it:1stcond} holds, then Condition \ref{it:2ndcond} is decidable.
\end{theorem}

The following lemma gives the last piece of the proof of \cref{thm:main}, as we shall explain after its proof.
\begin{lemma}
\label{lem: sync seq}
  Let $f:A^*\to A^*$ be a Parikh-collinear morphism prolongable on $a$. For all $b\in A$, the sequence $(|\pref_n(\infw{x})|_b)_{n\ge 0}$ is $k$-synchronized.
\end{lemma}

\begin{proof}
Let $b\in A$. Since $f$ is Parikh-collinear, for each immortal letter $c$,  the ratio $|f(c)|_b/|f(c)|$ is constant and depends only on~$b$.  Thus write $|f(a)|_b/|f(a)|=r/q$. 
  
Consider a prefix of $\infw{x}$ of the form $\pref_n(\infw{x})=f(x_n)t_n$ where $x_n$ is a prefix of $\infw{x}$ such that $\PR(n) = |f(x_n)|$.  Since $|f(x_n)|_b = \frac{r}{q}|f(x_n)|$,  we get $q |\pref_n(\infw{x})|_b = r |f(x_n)| + q|t_n|_b$. Define the function $F(n) = |\pref_n(\infw{x})|_b$ for all $n\ge 0$.  Then
\[
y = F(n) \equiv \exists m,z \colon (\PR(n) = m) \wedge (q\cdot (y-z) = r \cdot m) \wedge (|\infw{x}[m...n-1]|_b = z).
\]
Since $|\infw{x}[\PR(n)...n-1]|$ attains finitely many values (recall that $\infw{x}$ has bounded abelian complexity \cite{CassaigneRichomeSaariZamboni2011}), the last check $(|\infw{x}[m...n-1]|_b = z)$ can be expressed by a first-order logical formula with indexing into $\infw{x}$. The formula $y=F(n)$ has two free variables, hence \cite[Thm.~10.2.3]{Shallit2022logical} asserts that there is a DFA recognizing the language $\{(n, F(n))_k : n\ge 0\}$, and is thus synchronized.
\end{proof}

As a corollary of \cite[Thm.~11]{CassaigneRichomeSaariZamboni2011}, the abelian complexity function of a fixed point of a Parikh-collinear morphism is bounded
by a constant, so Condition \ref{it:2ndcond} in \cref{thm:shallit} is satisfied. Since Condition \ref{it:1stcond} is equivalent to the property
that for each $a \in A$, the sequence $(|\pref_n(\infw{x})|_a)_{n\geq 0}$ is synchronized, the above lemma allows to conclude with the proof of 
\cref{thm:main}.

\section{A detailed discussion of the procedure}\label{sec:4}

Throughout this section, we let $f$ be defined by $0 \mapsto 012,1\mapsto 1120022 \mapsto \varepsilon$,
and $f^{\omega}(0)=\infw{x} =x_0x_1\cdots$. Our aim is to prove the following.

\begin{proposition}
  The fixed point $\infw{x}=012112002112002\cdots$ of the Parikh-collinear morphism $f\colon 0 \mapsto 012, 1\mapsto 112002,2 \mapsto \varepsilon$  has abelian complexity equal to $135(377)^{\omega}$.
\end{proposition}

Computing $\sum_{a=0}^2 |f(a)|_a = 3$, we know that $\infw{x}$ is $3$-automatic.
In \cite{RSW-DLT23} with \cref{thm:kautomatic}, we give an effective procedure to compute an equivalent morphic representation;
the procedure produces the coding $\tau$ defined by
\[
\widehat{0}_1,\widehat{1}_4,\widehat{1}_5 \mapsto 0;\quad
\widehat{0}_2,\widehat{1}_1,\widehat{1}_2 \mapsto 1;\quad
\widehat{0}_3,\widehat{1}_3,\widehat{1}_6 \mapsto 2
\]
and the $3$-uniform morphism $g$ defined by
\begin{align*}
\widehat{0}_1,\widehat{1}_5,\widehat{1}_6 \mapsto \widehat{0}_1\widehat{0}_2\widehat{0}_3;\quad
\widehat{0}_2, \widehat{1}_1,\widehat{1}_3 \mapsto \widehat{1}_1\widehat{1}_2\widehat{1}_3;\quad
\widehat{0}_3, \widehat{1}_2, \widehat{1}_4 \mapsto \widehat{1}_4\widehat{1}_5\widehat{1}_6,
\end{align*}
so that $\tau(g^{\omega}(\widehat{0}_1)) = \infw{x}$.

One notes that there are redundant letters (i.e.,  they have equal images under both $\tau$ and $g \circ \tau$). We thus find a simpler morphism $h$
by identifying them:
\[
0 \mapsto 0 1 2;  \quad
1  \mapsto 1 3 4; \quad
2 \mapsto 5 0 6; \quad
3 \mapsto 5 0 6; \quad
4 \mapsto 1 3 4; \quad
5 \mapsto 5 0 6; \quad
6 \mapsto 0 1 2,
\]
with which $\tau'(h^{\omega}(0)) = \infw{x}$, where $\tau'$ is defined by $0,5\mapsto 0$; $1,3 \mapsto 1$; $2,4,6 \mapsto 2$.

We may introduce the $3$-automatic word $\infw{x}$ to {\tt Walnut} as follows:

\begin{verbatim}
morphism h "0->012 1->134 2->506 3->506 4->134 5->506 6->012";
morphism tau "0->0 1->1 2->2 3->1 4->2 5 ->0 6->2";
promote H h;
image X tau H;
\end{verbatim}

{\tt Walnut} now knows the infinite word as \verb|X|, and it is now easy to verify that $\infw{x}$ is aperiodic. Indeed, \cref{eq:periodic} translates to:
\begin{verbatim}
eval isaperiodic "?msd_3 ~(Ep,i p>0 & (An n>i => (X[n]=X[n+p])))";
\end{verbatim}
and {\tt Walnut} produces \verb|True|.

Following the procedure, we next wish to compute (or bound) the constant of recognizability.
The bound given in \cref{cor:cnstnt-rec} is $ (6^{36}+1) \cdot 6^{7\cdot(2+6^{27})+2}$, which is unmanageable in practice. 
At this point, we compute the actual constant of recognizability (with the help of \verb|Walnut|) to proceed with the illustration.

\begin{lemma}
\label{lem: constant rec f 2}
Given $f\colon 0 \mapsto 012,1\mapsto 1120022 \mapsto \varepsilon$,  its constant of recognizability is $2$.
\end{lemma}
\begin{proof}
We observe that the factor $120$ appears both in $f(00)=012012$ and $f(1)=112002$, so the pair $(1,20)$ is not synchronizing. This observation
bounds the constant of recognizability from below by $2$.

We observe that each factor of length $5$ contains at least one occurrence of $2$; this is because $2$ appears at the position $n$ if and only if $n\equiv 2 \bmod{3}$,
a fact that can be verified using \verb|Walnut|:

\begin{verbatim}
eval appearance2 "?msd_3 An X[n]=@2 <=> Em n=3*m+2";
\end{verbatim}

Let $w = uv$ be a factor of $\infw{x}$ with $|u|=2$ and $|v|=3$. From the above we deduce that $v$ contains an occurrence of $2$.

Since any cutting point is either $0$ or appears just after an occurrence of $2$ (both $f(0)$ and $f(1)$ end with $2$), it suffices to inspect the two letters appearing just before a $2$ in $v$. Indeed, the return words to $2$ in $\infw{x}$ are $201$, $200$, and $211$;
if the two preceding letters are $00$ or $01$, then the position after $2$ is a cutting point. Otherwise it is not.
Thus the constant of recognizability is bounded above by $2$.
\end{proof}

We next proceed to define the cutting sequence of $\infw{x}$ in {\tt Walnut} as follows;
the index $n$ is a cutting point if $n=0$ or $n\geq 3$ and $x_{n-1}=2$ and $x_{n-3}\neq 1$ (this can be deduced from the proof of~\cref{lem: constant rec f 2}).
\smallskip
\begin{verbatim}
def cut "?msd_3 n=0 | (n>=3 & X[n-1]=@2 & ~(X[n-3]=@1))";
\end{verbatim}
\smallskip

Using the \verb|cut| set, we define the pairs $(n,x)$ such that $x$ is the largest cut point that is at most $n$.
The following predicate \verb|prev| recognizes exactly these pairs (see \cref{lem: ne and pr}).
\begin{verbatim}
def prev "?msd_3 x<=n & $cut(x) & (Ay (y>x & y<=n)=>~$cut(y))";
\end{verbatim}
\smallskip
Next, we define the synchronized sequence of the number of $0$'s (resp., $1$'s,  $2$'s) in the prefix of length $n+1$, $n\geq 0$.
\smallskip
\begin{verbatim}
def prefn0 "?msd_3 (n<=2 & y=1) | (3<=n & Em,z ($prev(n,m) & 3*y=m+3*z
     & ((X[m]=@0 & z=1) |
     (X[m]=@1 & ((n<m+3 & z=0) | (n=m+3 & z=1) | (n>=m+4 & z=2))))))";
\end{verbatim}
\smallskip
Here \verb|prefn0(n,y)| is true if $y$ is the number of $0$'s appearing in the prefix of length $n+1$ of $\infw{x}$, with $n\geq 0$.
We note that if $|f(w)| = \ell$, then $|f(w)|_0 = \ell/3$. Hence for a prefix $f(w)z$, where $z$ is a prefix of the image of a letter, we have $|f(w)z|_0 = \ell/3 + |z|_0$.
If $z$ begins with $0$, then we know that $z$ is a prefix of $f(0)$, so that $|z|_0 = 1$ as long as $z\neq \varepsilon$.
Otherwise $z$ is a prefix of $f(1)$, and $|z|_0 = 0$ if $|z|\leq 3$; $|z|_0 = 1$ if $|z|=4$; and $|z|_0 =2$ otherwise.

With similar arguments one can see that the following predicates define the pairs $(n,|\pref_{n+1}(\infw{x})|_a)$, for $a\in\{1,2\}$.
\smallskip
\begin{verbatim}
def prefn1 "?msd_3 Em,z $prev(n,m) & 3*y=m+3*z &
     ((X[m]=@0 & ((m=n & z=0) | (n>=m+1 & z=1))) |
     (X[m]=@1 & ((m=n & z=1) | (n>=m+1 & z=2))))";

def prefn2 "?msd_3 Em,z $prev(n,m) & 3*y=m+3*z &
     ((X[m]=@0 & ((n<m+2 & z=0) | (m+2=n & z=1))) | (X[m]=@1 &
     ((n<m+2 & z=0) | (n>=m+2 & n<m+5 & z=1) | (n=m+5 & z=2))))";
\end{verbatim}

From this point onward we may proceed as outlined by Shallit in \cite{Shallit2021abelian} to find the abelian complexity function of $\infw{x}$ as a $3$-automatic sequence.

\begin{remark}
We could now find the sequence $(\Psi(\pref_n(\infw{x})))_{n\geq 0}$ as a synchronized sequence
$(n,\Psi(\pref_n(x_n)))_{n\geq 0}$ with the following command:
\begin{verbatim}
def PrefParikhSync "?msd_3 (n=0 & x=0 & y=0 & z=0) |
     (n>=1 & $prefn0(n-1,x) & $prefn1(n-1,y) & $prefn2(n-1,z))";
\end{verbatim}
\noindent However, the automaton seems to be too complex to work with in a practical way.
\end{remark}

We shall opt to proceed with the approach of \cite{Shallit2021abelian} presented for the Tribonacci word. In particular, we work with the synchronized sequences
$(n,|\pref_n(\infw{x})|_a)_{n\geq 0}$, $a\in\{0,1,2\}$, separately instead:
\begin{verbatim}
def pref0 "?msd_3 (n=0 & y=0) | (n>=1 & $prefn0(n-1,y))";
def pref1 "?msd_3 (n=0 & y=0) | (n>=1 & $prefn1(n-1,y))";
def pref2 "?msd_3 (n=0 & y=0) | (n>=1 & $prefn2(n-1,y))";
\end{verbatim}

Next we define the automata accepting the triples $(i,n,|x_i\cdots x_{i+n-1}|_a)$, for $a\in\{0,1,2\}$:
\begin{verbatim}
def syncin0 "?msd_3 Ax,y ($pref0(i,x) & $pref0(i+n,y)) => (z + x = y)";
def syncin1 "?msd_3 Ax,y ($pref1(i,x) & $pref1(i+n,y)) => (z + x = y)";
def syncin2 "?msd_3 Ax,y ($pref2(i,x) & $pref2(i+n,y)) => (z + x = y)";
\end{verbatim}

For each $a\in\{0,1,2\}$, we inspect the possible differences $|v|-|\pref_n(\infw{x})|_a$, where $v$ ranges over
the factors of $\infw{x}$ of length $n$. Since $\infw{x}$ is guaranteed to have bounded abelian complexity,
there are only finitely many such possible values; here we may inspect the possible values of the differences as follows.
The first predicate accepts the non-negative values of $k$ such that $k=|v|_0-|\pref_n(\infw{x})|$; 
the second accepts the non-negative values of $k$ such that $-k = |v|_0-|\pref_n(\infw{x})|_0$:
\begin{verbatim}
def diffs0pos "?msd_3 En,i,x,z $syncin0(i,n,x) & $pref0(n,z) & x=k+z";
def diffs0neg "?msd_3 En,i,x,z $syncin0(i,n,x) & $pref0(n,z) & x+k=z";
\end{verbatim}
Inspecting the automata obtained, the possible values in both cases are $k=0$, $1$, $2$.
This implies that $\left||v|_0 - |\pref_n(\infw{x})|_0\right| \leq 2$, and all possible values are attained.
With similar inspections for the other letters, we find that
\[
-3\leq |v|_1 - |\pref_n(\infw{x})|_1 \leq 2 
\quad \text{and} \quad 
0\leq |v|_2 - |\pref_n(\infw{x})|_2 \leq 1,
\]
and all possible values are attained.
We thus have that
\begin{equation}\label{eq:possible-ranges}
\Psi(v)-\Psi(\pref_n(\infw{x})) \in \{-2,\ldots,2\} \times \{-3,\ldots,2\} \times \{0,1\}
\end{equation}
for any factor $v$ of length $n$. 
(Note that this in particular shows that $\infw{x}$ has bounded abelian complexity.)

If $S$ is the set of possible triples in~\cref{eq:possible-ranges},
we note that the triples $S + (2,3,0)$ will be non-negative. To get all the possible triples~$(s,t,u)$,
we inspect the automaton constructed by \verb|Walnut| with the command
\begin{verbatim}
def validtriples "?msd_3 Ei,n,a,b,c,d,e,f
     $syncin0(i,n,a) & $pref0(n,b) & s+b = a+2 &
     $syncin1(i,n,c) & $pref1(n,d) & t+d = c+3 &
     $syncin2(i,n,e) & $pref2(n,f) & u+f = e";
\end{verbatim}
\begin{figure}
\centering
\includegraphics[scale=.5,trim=35 35 90 620, clip]{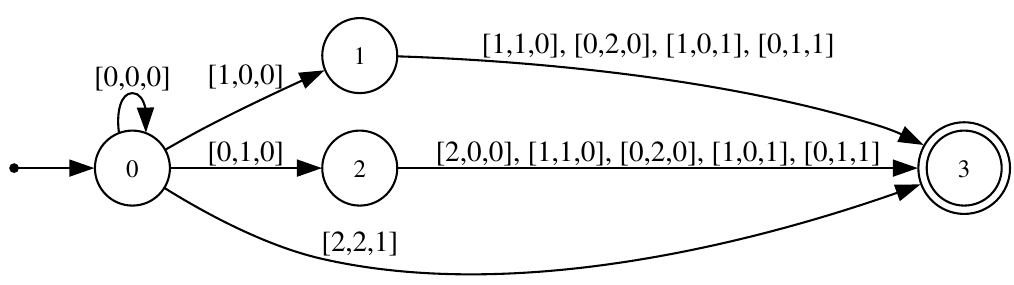}
\caption{The automaton accepting base-$3$ representations of triples of the form $\Psi(v)-\Psi(\pref_n(\infw{x})) + (2,3,0)$, where $v$ ranges through the factors of length $n$, and $n$ ranges through the natural numbers.}
\label{fig:aut for triples diff Parikh}
\end{figure}
We obtain the automaton in~\cref{fig:aut for triples diff Parikh}.
Inspecting it, we see that from the set appearing in \cref{eq:possible-ranges}, all ten vectors with the property that the entries sum to $0$
are attainable (we are comparing Parikh vectors of two words of the same length~$n$); they are
\begin{align*}
\{&(-2,2,0), (-2,1,1), (-1,1,0), (-1,0,1), (0,-1,1), \\
&(0,0,0), (1,-2,1), (1,-1,0), (2,-3,1), (2,-2,0)\}.
\end{align*}
Since $(0,0,0)$ is attained for any length $n$, we have $2^9$ possible sets of difference vectors to consider.
For each such difference vector, we may provide a predicate recognizing those $i$ and $n$ for which the vector is attained.
For example, the vector $(-2,2,0)$ is defined with the command
\begin{verbatim}
def vecn220 "?msd_3 Ea,b,c,d,e,f
     $syncin0(i,n,a) & $pref0(n,b) & a+2=b &
     $syncin1(i,n,c) & $pref1(n,d) & c=d+2 &
     $syncin2(i,n,e) & $pref2(n,f) & e=f";
\end{verbatim}

In principle, one could then consider all $2^9$ possible combinations of possible difference vectors
for a given length $n$. However, computations suggest that the possible sets of vectors are the following:
\begin{align*}
S_1 &= \{(-1, 0, 1), (-1, 1, 0), (0, 0, 0)\},\\
S_2 &= \{(-1, 0, 1), (-1, 1, 0), (0, -1, 1), (0, 0, 0), (1, -1, 0)\},\\
S_3 &= \{(-1, 1, 0), (0, 0, 0), (1, -1, 0)\},\\
S_4 &= \{(-1, 0, 1), (-1, 1, 0), (0, -1, 1), (0, 0, 0), (1, -2, 1), (1, -1, 0), (2, -2, 0)\},\\
S_5 &= \{(-1, 0, 1), (0, -1, 1), (0, 0, 0), (1, -2, 1), (1, -1, 0), (2, -3, 1), (2, -2, 0)\},\\
S_6 &= \{(0, 0, 0), (1, -1, 0), (2, -2, 0)\},\\
S_7 &= \{(-2, 1, 1), (-1, 0, 1), (-1, 1, 0), (0, -1, 1), (0, 0, 0), (1, -2, 1), (1, -1, 0)\},\\
S_8 &= \{(-2, 1, 1), (-2, 2, 0), (-1, 0, 1), (-1, 1, 0), (0, -1, 1), (0, 0, 0), (1, -1, 0)\}.
\end{align*}

Let us define, for each $i$, the lengths $n$ for which the set $S_i$ is attained.
For example, the lengths corresponding to $S_1$, $S_2$, $S_3$, and $S_6$ would be defined as\\
\verb~def S1 "?msd_3 Ai ($vecn101(i,n) | $vecn110(i,n) | $vec000(i,n))~\\
\verb~                              & (Ej,k $vecn101(j,n) & $vecn110(k,n))";~\\
\verb~def S2 "?msd_3 Ai ($vecn101(i,n) | $vecn110(i,n) | $vec0n11(i,n)~\\
\verb~                                 | $vec000(i,n) | $vec1n10(i,n)) &~\\
\verb~         Ej,k,l,m ($vecn101(j,n) & $vecn110(k,n) & $vec0n11(l,n)~\\
\verb~            & $vec1n10(m,n))";~\\
\verb~def S3 "?msd_3 (Ai $vecn110(i,n) | $vec000(i,n) | $vec1n10(i,n)) &~\\
\verb~           (Ej,k $vecn110(j,n) & $vec1n10(k,n))";~\\
\verb~def S6 "?msd_3 (Ai $vec000(i,n) | $vec1n10(i,n) | $vec2n20(i,n)) &~\\
\verb~           (Ej,k $vec1n10(j,n) & $vec2n20(k,n))";~

To avoid memory issues, we split the definitions of $S_4$, $S_5$, $S_7$,  and $S_8$
into several parts:\\
\verb~def S4a "?msd_3 Ai $vecn101(i,n) | $vecn110(i,n) | $vec0n11(i,n) |~\\
\verb~       $vec000(i,n) | $vec1n21(i,n) | $vec1n10(i,n) | $vec2n20(i,n)";~
\verb~def S4b "?msd_3 $S4a(n) & Ei $vecn101(i,n)";~\\
\verb~def S4c "?msd_3 $S4b(n) & Ei $vecn110(i,n)";~\\
\verb~def S4d "?msd_3 $S4c(n) & Ei $vec0n11(i,n)";~\\
\verb~def S4e "?msd_3 $S4d(n) & Ei $vec1n21(i,n)";~\\
\verb~def S4f "?msd_3 $S4e(n) & Ei $vec1n10(i,n)";~\\
\verb~def S4 "?msd_3 $S4f(n) & Ei $vec2n20(i,n)";~\\
\newline
\verb~def S5a "?msd_3 Ai ($vecn101(i,n) | $vec0n11(i,n) | $vec000(i,n) |~\\
\verb~   $vec1n21(i,n) | $vec1n10(i,n) | $vec2n31(i,n) | $vec2n20(i,n))";~\\
\verb~def S5b "?msd_3 $S5a(n) & Ei $vecn101(i,n)";~\\
\verb~def S5c "?msd_3 $S5b(n) & Ei $vec0n11(i,n)";~\\
\verb~def S5d "?msd_3 $S5c(n) & Ei $vec1n21(i,n)";~\\
\verb~def S5e "?msd_3 $S5d(n) & Ei $vec1n10(i,n)";~\\
\verb~def S5f "?msd_3 $S5e(n) & Ei $vec2n31(i,n)";~\\
\verb~def S5 "?msd_3 $S5f(n) & Ei $vec2n20(i,n)";~\\
\newline
\verb~def S7a "?msd_3 Ai $vecn211(i,n) | $vecn101(i,n) | $vecn110(i,n) |~\\
\verb~     $vec0n11(i,n) | $vec000(i,n) | $vec1n21(i,n) | $vec1n10(i,n)";~\\
\verb~def S7b "?msd_3 $S7a(n) & Ei $vecn211(i,n)";~\\
\verb~def S7c "?msd_3 $S7b(n) & Ei $vecn101(i,n)";~\\
\verb~def S7d "?msd_3 $S7c(n) & Ei $vecn110(i,n)";~\\
\verb~def S7e "?msd_3 $S7d(n) & Ei $vec0n11(i,n)";~\\
\verb~def S7f "?msd_3 $S7e(n) & Ei $vec1n21(i,n)";~\\
\verb~def S7 "?msd_3 $S7f(n) & Ei $vec1n10(i,n)";~\\
\newline
\verb~def S8a "?msd_3 Ai $vecn211(i,n) | $vecn220(i,n) | $vecn101(i,n) |~\\
\verb~     $vecn110(i,n) | $vec0n11(i,n) | $vec000(i,n) | $vec1n10(i,n)";~\\
\verb~def S8b "?msd_3 $S8a(n) & Ei $vecn211(i,n)";~\\
\verb~def S8c "?msd_3 $S8b(n) & Ei $vecn220(i,n)";~\\
\verb~def S8d "?msd_3 $S8c(n) & Ei $vecn101(i,n)";~\\
\verb~def S8e "?msd_3 $S8d(n) & Ei $vecn110(i,n)";~\\
\verb~def S8f "?msd_3 $S8e(n) & Ei $vec0n11(i,n)";~\\
\verb~def S8 "?msd_3 $S8f(n) & Ei $vec1n10(i,n)";~

To obtain the abelian complexity function as an automatic sequence, we finally perform the following commands:\\
\verb~def abcomp1 "?msd_3 n=0";~\\
\verb~def abcomp3 "?msd_3 $S1(n) | $S3(n) | $S6(n)";~\\
\verb~def abcomp5 "?msd_3 $S2(n)";~\\
\verb~def abcomp7 "?msd_3 $S4(n) | $S5(n) | $S7(n) | $S8(n)";~\\
\verb~combine abcomp abcomp1=1 abcomp5=5 abcomp3=3 abcomp7=7;~\\
The first four automata recognize those lengths $n$ for which the abelian complexity equals $1$, $3$, $5$, and $7$, respectively.
The last combines these automata to form an automatic sequence over the alphabet
$\{1,3,5,7\}$. The automaton obtained is depicted in \cref{fig:abelian_complexity}.
\begin{figure}
\centering
\includegraphics[scale=.15,]{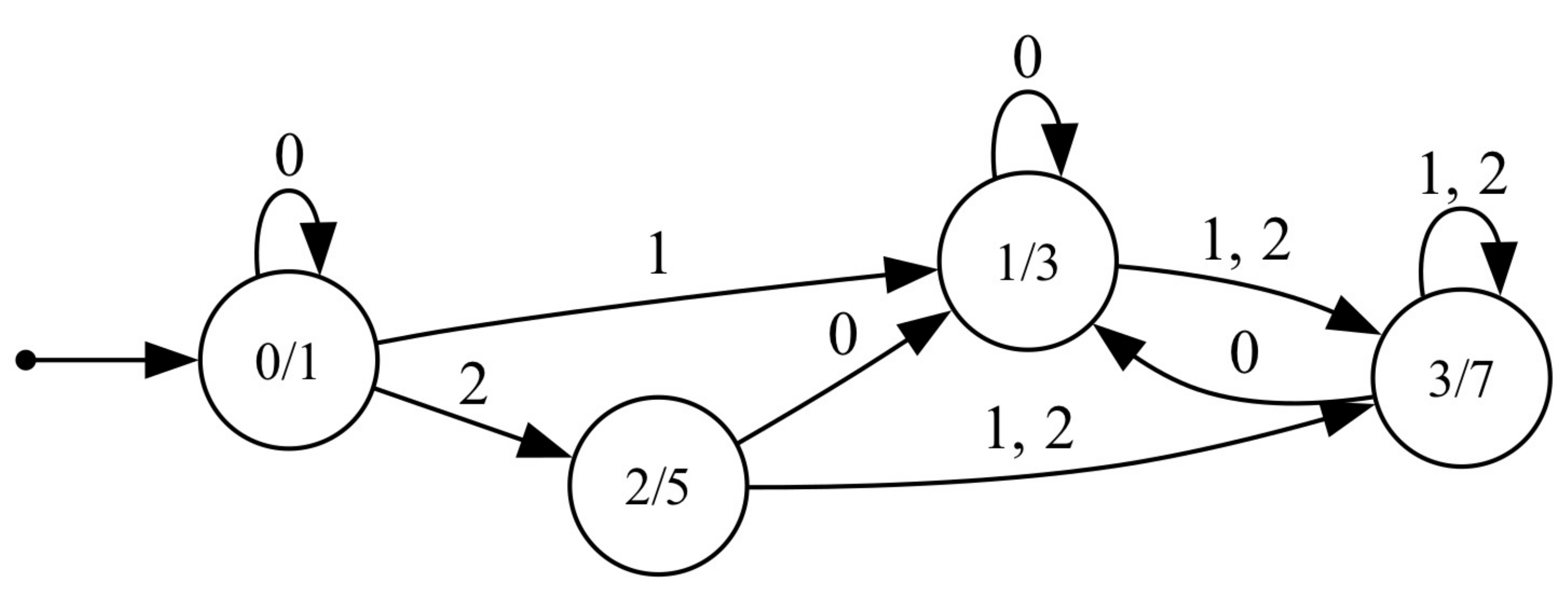}
\caption{The abelian complexity function of the fixed point $\infw{x}=012112002112002\cdots$ of the Parikh-collinear morphism $f\colon 0 \mapsto 012, 1\mapsto 112002,2 \mapsto \varepsilon$ as a $3$-automatic sequence.}
\label{fig:abelian_complexity}
\end{figure}
Inspecting the automaton, we see that the abelian complexity function of $\infw{x}$
equals $135(377)^{\omega}$,  as desired.

\section{Concluding remarks}\label{sec13}

We may address similar questions. In the same vein as Sportiello and Salo's question, we may ask: Is the abelian complexity of the fixed point of any {\em Parikh-constant morphism}, i.e., all images of letters have the same Parikh vector \cite{RigoSalimov2015}, always ultimately periodic? We know with \cite{RSW-DLT23} that this is not the case for an arbitrary Parikh-collinear morphism, but Parikh-constant morphisms are more restrictive: all columns of $M_f$ are the same.

In \cref{thm:more_gen}, there is an assumption about recognizability. Nevertheless, there are situations where recognizability does not hold and still, the cutting set is definable. As an example, consider the non-uniform morphism $f:a\mapsto ab$, $b\mapsto b'c'$, $b'\mapsto b$, $c'\mapsto ccc$ and $c\mapsto cc$. Its fixed point starting with $a$ is also $2$-automatic and generated by the morphism $a\mapsto ab$, $b\mapsto b'c'$, $b'\mapsto bc$, $c'\mapsto cc$, $c\mapsto cc$ (a slight modification of the morphism used to generate the characteristic sequence of powers of $2$). Because of the arbitrarily long blocks of $c$'s appearing in $f^\omega(a)$, $f$ is not recognizable on this infinite word. Nevertheless, the cutting set $\mathsf{CS}_{f,a}=\{0,2,4,5,8,10,\ldots\}$ is $2$-definable because it is easy to see that it is of the form  $(2\mathbb{N}\setminus\{4^n+2\mid n>0\})\cup \{4^n+1\mid n>0\}$.
So the conclusion of \cref{thm:more_gen} may hold for a larger class of morphic words (being simultaneously $k$-automatic for some~$k$). 

\backmatter

\bmhead{Acknowledgments}

We thank J.~Leroy for fruitful discussions on morphic words and pointing out useful references. 
We thank A.~Sportiello and V.~Salo for asking the question leading to this paper.
We warmly thank M.-P.~B\'eal,  F.~Durand, and D.~Perrin for sharing a draft of their book~\cite{BDP2024}.
We also thank the reviewers of~\cite{RSW-DLT23} for their suggestions.

\section*{Declarations}

\begin{itemize}
\item Funding \\
M.~Rigo is supported by the FNRS Research grant T.0196.23 (PDR).
M.~Stipulanti is an FNRS Research Associate supported by the Research grant 1.C.104.24F.
M.~Whiteland is supported by the FNRS Research grant 1.B.466.21F.
\item Conflict of interest/Competing interests \\
The authors have no relevant interests to disclose.
\item Ethics approval and consent to participate \\
Not applicable
\item Consent for publication \\
Not applicable
\item Data, materials,  and code availability \\ 
Not applicable
\item Author contribution \\
All authors equally contributed to the main content.
\end{itemize}

%
%
%
%
%
%
%
%
%

\bibliography{bibliography}

\end{document}